\documentclass[sigconf,nonacm]{acmart}

\usepackage{tabularx}
\usepackage{stfloats}
\usepackage{amsmath}
\usepackage{bm}
\usepackage{nicefrac}
\usepackage{booktabs}
\usepackage{array}
\usepackage{multirow}
\usepackage{threeparttable}
\usepackage{makecell}
\usepackage[procnumbered,ruled,vlined,linesnumbered]{algorithm2e}
\usepackage{siunitx}
\usepackage{stfloats}
\usepackage{graphicx}
\usepackage{subfigure}
\usepackage{hyperref}
\usepackage{array}
\usepackage{epstopdf}
\usepackage{mweights}
\usepackage{balance}
\usepackage{adjustbox}
\usepackage{numprint}

\newcommand{\rea}{\mathbb{R}}

\newcommand{\one}{\mathbf{1}}
\newcommand{\zero}{\mathbf{0}}

\newcommand\XX{\boldsymbol{\mathit{X}}}

\newcommand\zz{\boldsymbol{\mathit{z}}}

\newcommand\aaa{\boldsymbol{\mathit{a}}}

\newcommand\bb{\boldsymbol{\mathit{b}}}
\newcommand\cc{\boldsymbol{\mathit{c}}}
\newcommand\dd{\boldsymbol{\mathit{d}}}
\newcommand\ee{\boldsymbol{\mathit{e}}}
\newcommand\sss{\boldsymbol{\mathit{s}}}

\newcommand\rr{\boldsymbol{\mathit{r}}}
\newcommand\hh{\boldsymbol{\mathit{h}}}

\renewcommand\AA{\boldsymbol{\mathit{A}}}
\newcommand\BB{\boldsymbol{\mathit{B}}}
\newcommand\CC{\boldsymbol{\mathit{C}}}

\newcommand\DD{\boldsymbol{\mathit{D}}}

\newcommand\PP{\boldsymbol{\mathit{P}}}

\newcommand\II{\boldsymbol{\mathit{I}}}

\def\trace#1{\mathrm{Tr} \left(#1 \right)}

\newcommand{\Polarization}{P(\mathcal{G})}
\newcommand{\Disagreement}{D(\mathcal{G})}
\newcommand{\Internal}{I(\mathcal{G})}
\newcommand{\Controversy}{C(\mathcal{G})}

\def\calG{\mathcal{G}}

\newcommand\LL{\bm{\mathit{L}}}

\SetKwInOut{Input}{Input\ \ \ \ }
\SetKwInOut{Output}{Output}

\begin{document}

\title{Efficient Algorithms for Relevant Quantities of Friedkin-Johnsen Opinion Dynamics Model}

\author{Gengyu Wang}
\affiliation{%
  \department{Computation  and Artificial Intelligence Innovative College}
  \institution{Fudan University}
  \city{Shanghai}
  \country{China}
}
\email{24210240091@m.fudan.edu.cn}

\author{Runze Zhang}
\affiliation{%
  \department{Computation  and Artificial Intelligence Innovative College}
  \institution{Fudan University}
  \city{Shanghai}
  \country{China}
}
\email{24210240410@m.fudan.edu.cn}

\author{Zhongzhi Zhang}
\affiliation{%
  \department{Computation  and Artificial Intelligence Innovative College}
    \institution{Fudan University}
    \city{Shanghai}
    \country{China}
}
\email{zhangzz@fudan.edu.cn}
\authornote{Corresponding author.}

\begin{abstract}
    Online social networks have become an integral part of modern society, profoundly influencing how individuals form and exchange opinions across diverse domains ranging from politics to public health. The Friedkin-Johnsen model serves as a foundational framework for modeling opinion formation dynamics in such networks. In this paper, we address the computational task of efficiently determining the equilibrium opinion vector and associated metrics including polarization and disagreement, applicable to both directed and undirected social networks. We propose a deterministic local algorithm with relative error guarantees, scaling to networks exceeding ten million nodes. Further acceleration is achieved through integration with successive over-relaxation techniques, where a relaxation factor optimizes convergence rates. Extensive experiments on diverse real-world networks validate the practical effectiveness of our approaches, demonstrating significant improvements in computational efficiency and scalability compared to conventional methods.
\end{abstract}

\keywords{Opinion dynamics, social network, multi-agent system, graph algorithm, Laplacian solver}

\maketitle

\section{Introduction}
Social media platforms have become an indispensable part of modern life, profoundly transforming how individuals disseminate, exchange, and form opinions~\cite{anderson2019recent,dong2018survey,jia2015opinion}. Compared to traditional methods of information dissemination, this shift has introduced unprecedented complexity, particularly in the dynamics of information propagation, significantly influencing human behavior~\cite{Le20}. The pervasive impact of online interactions has sparked considerable academic interest in understanding the mechanisms governing opinion propagation and formation. To this end, researchers have developed a variety of mathematical models to describe opinion dynamics ~\cite{jia2015opinion, PrTe17, dong2018survey, anderson2019recent}, among which the Friedkin-Johnsen (FJ) model ~\cite{FrJo90} has emerged as one of the most widely applied frameworks. Notably, adaptations of the FJ model have been used to analyze complex negotiation processes, such as those leading to the Paris Agreement, demonstrating its utility in explaining consensus formation in multilateral international settings~\cite{BeWaVaHoShAl21}.

A central focus in opinion dynamics is the analysis of equilibrium opinions, which represent individuals' expressed opinions at steady state. These equilibrium values enable the study of opinion distributions and the quantification of key social phenomena. Various measures, including disagreement~\cite{dandekar2013biased, musco2018minimizing}, polarization~\cite{matakos2017measuring,dandekar2013biased, musco2018minimizing}, conflict~\cite{chen2018quantifying}, and controversy~\cite{chen2018quantifying}, have been developed to characterize these dynamics. 
As network sizes grow, efficiently computing equilibrium opinion vectors has become a significant challenge. 
Traditional matrix inversion methods, while exact, incur a time complexity of $O(n^3)$, where $n$ represents the number of nodes in the network. This renders them infeasible for large-scale networks. Existing approaches primarily rely on Laplacian solver~\cite{xu2021fast}, which achieve sublinear time complexity but are limited to undirected graphs. Recent work~\cite{neumann2024sublinear} proposes lazy random walk-based methods to estimate single-node opinions and related measures without computing the full equilibrium vector. However, these methods face challenges in scalability or error guarantees, limiting their broader applicability. Another method is based on forest sampling, which extends Wilson's algorithm to generate uniform spanning trees, thereby efficiently estimating equilibrium opinion vectors in directed graphs~\cite{sunOpinionOptimizationDirected2023} . Although this method demonstrates good scalability for large networks, its sampling process remains computationally expensive.

To address these challenges, we propose a novel local iteration-based algorithm that efficiently approximates the equilibrium opinion vector in the FJ model while guaranteeing relative error. Additionally, we introduce a method to enhance the algorithm's robustness when nodes have zero internal opinions. Furthermore, we incorporate successive over-relaxation (SOR) techniques to significantly accelerate the convergence process and improve computational efficiency. Unlike most algorithms that can only be applied to undirected graphs, our algorithms support both directed and undirected networks, while avoiding time-consuming sampling work in forest sampling algorithm. Experimental results demonstrate that our proposed algorithms exhibit superior performance and scalability on networks comprising tens of millions of nodes.
The main contributions of this paper are summarized as follows:
\begin{itemize}
    \item We introduce a robust local iteration algorithm that efficiently approximates the equilibrium opinion vector while guaranteeing relative error bounds.
    \item By incorporating SOR techniques, we optimize the algorithm with a relaxation parameter $\omega$, leading to faster convergence and improved computational efficiency.
    \item We conduct extensive numerical experiments on real-world network datasets. The results validate the advantages of our algorithms in terms of high efficiency and strong scalability.
\end{itemize}

The remainder of this paper is organized as follows: Section~\ref{sec:related} reviews related work on opinion dynamics and the FJ model; Section~\ref{sec:prelim} introduces notation and preliminaries necessary for understanding the subsequent sections; Section~\ref{sec:problem} provides an overview of existing algorithms and their limitations; Section~\ref{sec:localIter} details the local iteration algorithm and SOR optimization; Section~\ref{sec:experim} presents experimental results and performance comparisons; and Section~\ref{sec:conclusion} concludes the paper and summarizes the key findings.

\section{Related Work}\label{sec:related}
The FJ model, integrating internal opinions with social influence, stands as a pivotal framework in social dynamics research~\cite{FrJo90}. Stability conditions for this model were established in~\cite{ravazzi2014ergodic}, while equilibrium expressed opinion was derived in~\cite{BiKlOr15,das2013debiasing}. Interpretations of the FJ model were further explored in~\cite{BiKlOr15} and~\cite{ghaderi2014opinion}. Research has also delved into adversarial interventions~\cite{chen2021adversarial,gaitonde2020adversarial,tu2023adversaries} and the spread of viral content~\cite{tu2022viral}. Furthermore, the FJ model has been expanded in various dimensions, as discussed in~\cite{jia2015opinion,semonsen2018opinion}, including multidimensional approaches~\cite{friedkin2016network,parsegov2016novel}.

Beyond the foundational aspects, interpretations, and extensions of the FJ model, researchers have developed metrics to quantify disagreement~\cite{musco2018minimizing}, polarization~\cite{matakos2017measuring,musco2018minimizing}, conflict~\cite{chen2018quantifying}, and controversy~\cite{chen2018quantifying}. These metrics offer valuable insights into social dynamics. Various optimization problems based on the FJ model have been addressed, including efforts to minimize overall opinion~\cite{sunOpinionOptimizationDirected2023}, disagreement~\cite{musco2018minimizing}, polarization~\cite{musco2018minimizing, matakos2017measuring}, and conflict~\cite{chen2018quantifying,zhu2022nearly} by adjusting node attributes such as internal opinions or susceptibility to persuasion~\cite{abebe2018opinion, gionis2013opinion,chan2019revisiting}. Network modifications, like edge additions or deletions to reduce polarization or foster consensus, have emerged as a significant area of study~\cite{zhu2021minimizing,musco2018minimizing, chen2018quantifying}.

Scalability and efficiency pose significant challenges in large networks, where conventional methods like matrix inversion are impractical. To overcome this, scalable algorithms utilizing sampling and distributed computing have been introduced~\cite{ravazzi2014ergodic, avena2018two}, facilitating the analysis of networks with millions of nodes. For undirected graphs, Laplacian solvers can approximate equilibrium opinions in near-linear time~\cite{xu2021fast}, while directed graphs rely on forest sampling methods~\cite{sunOpinionOptimizationDirected2023}. Recent work~\cite{neumann2024sublinear} propose sublinear time algorithms to estimate single-node opinions without computing the entire equilibrium vector.

\section{Preliminaries}\label{sec:prelim}
This section is devoted to a brief introduction to some useful notations and tools, in order to facilitate the description of  problem formulation and  algorithms.

\subsection{Notations}
We denote scalars in $\rea$ by normal lowercase letters like $a,b,c$, sets by normal uppercase letters like $A,B,C$, vectors by bold lowercase letters like $\aaa, \bb, \cc$, and matrices by bold uppercase letters like $\AA, \BB, \CC$. We use $\one$ to denote  the vector  of appropriate dimensions with all entries being ones, and use $\ee_i$  to denote the $i^{\rm th}$ standard basis vector of appropriate dimension. Let $\aaa^\top$ and $\AA^\top$  denote, respectively, transpose of  vector $\aaa$ and matrix  $\AA$. Let $\trace \AA$ denote the trace of matrix $\AA$. We write $a_{i,j}$ to denote the entry at row $i$ and column $j$ of $\AA$ and we use $a_i$ to denote the $i^{\rm th}$ element of vector $\aaa$.  
Let $a_{\max}$, $a_{\min}$, $\bar{a}$ and $a_{\mathrm{sum}}$ denote the maximum element, minimum element, average value of vector 
$\aaa$ and the sum of all elements in vector $\aaa$, respectively.

\subsection{Graphs and Related Matrices}

Let $\calG = (V,E)$ be an unweighted simple undirected graph with $n = |V|$ nodes and  $m = |E|$ edges, where $V$ is the set of nodes and $E \subseteq V \times V$ is the set of edges. The adjacency relation of all nodes in $\calG$ is encoded in the adjacency matrix $\AA$,  whose entry $a_{i, j}=1$ if $i$ and $j$ are adjacent, and $a_{i, j}=0$ otherwise. For any given node $i$, $N_i$ denote the  set  of its neighbors, meaning $N_i =\{ j: (i,j)\in E\}$. The degree \(d_i\) of any node \(i\) is defined by \(d_i=\sum_{j=1}^n a_{ij}\). The degree diagonal matrix of $\calG$ is $\DD={\rm diag}(d_1, d_2, \cdots, d_n)$. The degree vector $\dd$ is given by $\dd = \DD\one$. The transition matrix of an unbiased random walk on graph $\calG$ is $\PP = \DD^{-1} \AA$, which is a row-stochastic matrix. The Laplacian matrix $\LL$ of $\calG$ is $\LL = \DD - \AA$, which is a symmetric semi-positive definite matrix. Matrix $\LL$ is singular, and thus it is irreversible. Moreover, for a nonnegative diagonal matrix $\XX$ with at least one nonzero diagonal entry, we have the fact that every element of $(\LL + \XX)^{-1}$ is positive~\cite{li2019current}.

\subsection{Friedkin-Johnsen Model}

The Friedkin-Johnsen (FJ) model~\cite{FrJo90} is a popular model  for opinion evolution and formation. For the FJ opinion model on a graph $\calG=(V,E)$, each node/agent  $i\in V$ is associated with two opinions: one is the internal opinion $s_i$, the other is the expressed opinion $z_i^{(t)}$ at time $t$. The internal opinion $s_i$ is in the closed interval $[0, 1]$, reflecting the intrinsic position of node $i $ on a certain topic, where 0 and 1 are polar opposites of opinions regarding the topic. A higher value of $s_i$ signifies that node $i$ is more favorable toward the topic, and vice versa. During the process of opinion evolution, the internal opinion $s_i$ remains constant, while the expressed opinion $z_i^{(t)}$ evolves at time $ t+1$ as
\begin{equation}\label{FJ}
z_i^{(t+1)} = \frac{s_i +\sum_{j\in N_i} a_{ij}z_j^{(t)}}{1+\sum_{j\in N_i} a_{ij}}.
\end{equation}
 Let $ \sss = (s_1,s_2,\ldots,s_n)^\top$ denote the vector of internal opinions, and let $ \zz^{(t)} = (z_1^{(t)},z_2^{(t)},\cdots,z_n^{(t)})^\top $ denote the vector of expressed opinions at time $ t $. 
\begin{lemma}[\cite{BiKlOr15}]\label{le-z}
	As $ t $ approaches infinity, $ \zz^{(t)} $ converges to an equilibrium vector  $ \zz = (z_1,z_2,\cdots,z_n)^\top$ satisfying $ \zz = (\II+\LL)^{-1}\sss $.
\end{lemma}
The matrix $(\II+\LL)^{-1}$ is the fundamental matrix of the FJ opinion dynamics model~\cite{gionis2013opinion}. It is a nonnegative and row-stochastic matrix, where each row sums to $1$. Moreover, the diagonal entry in each column is strictly larger than all other entries in that column~\cite{ChSh97,ChSh98}.

\subsection{Measures in FJ Model}

In the FJ model, opinions often fail to reach consensus, resulting in phenomena such as conflict, disagreement, polarization, and controversy. For a graph $\calG= (V,E)$, let $\tilde{\zz} = \zz - \frac{\zz^{\top} \textbf{1}}{n} \textbf{1}$ denote the mean-centered equilibrium vector. Table~\ref{tab:measures} summarizes key quantitative measures for these phenomena based on the FJ model.

\begin{table}[th]
\centering
\caption{Quantitative Measures in the FJ Model.}
\label{tab:measures}
\begin{tabular}{c c}
\toprule
\textbf{Measures} & \textbf{Definition} \\
\midrule
Internal Conflict & $\Internal =  \sum_{i\in V}\left(z_{i}-s_{i}\right)^{2}$ \\ 
Disagreement & $\Disagreement =   \sum_{(i,j) \in E} (z_i-z_j)^2$ \\ 
Polarization & $\Polarization =  \sum_{i \in V}\tilde{z} _i^2 =\tilde{\zz}^{\top} \tilde{\zz}$\\ 
Controversy & $\Controversy =  \sum_{i \in V}  z _i^2 =\zz^\top \zz$ \\ 
\bottomrule
\end{tabular}
\end{table}

Note that the disagreement $\Disagreement$ is also referred to as the external conflict of graph $\calG$ in~\cite{chen2018quantifying}. Additionally, the FJ model has been used to study the price of anarchy in societal opinion dynamics, where individuals minimize their stress, defined for node $i$ as $(z_i - s_i)^2 + \sum_{j \in N_i}(z_i-z_j)^2$. The total stress across all nodes is $\Internal + \Disagreement$, equivalent to the sum of internal and external conflict.

\section{Existing Algorithms}\label{sec:problem}
In this section, we introduce existing algorithms for computing the equilibrium opinion vector. Since direct matrix inversion is computationally infeasible due to its $O(n^3)$ complexity, alternative methods are necessary. Algorithms for directed graphs are discussed in Appendix~\ref{appdix:direct}. We present the comparison between algorithms in Table~\ref{tab:comp}.

\begin{table}[htb]
    \centering
    \caption{Comparison of Algorithms.}
    \label{tab:comp}
    \begin{tabular}{l c c}
    \toprule
    \textbf{Algorithms} & \textbf{Graph Type} & \textbf{Error Type}\\
    \midrule
    RWB~\cite{neumann2024sublinear} & undirected & absolute \\ 
    \textsc{Solver}~\cite{xu2021fast} & undirected & $l_2$ squared relative\\ 
    \textsc{Forest} & directed/undirected & absolute\\ 
    BLI/BLISOR (ours) & directed/undirected & relative\\ 
    \bottomrule
    \end{tabular}
\end{table}

\subsection{Algorithm Based on Random Walk}
    For the FJ model of opinion dynamics, \cite{neumann2024sublinear} introduces a sublinear-time algorithm to estimate the expressed opinion of a single node in large-scale networks. The algorithm called \textsc{RWB} uses lazy random walks starting from the target node and averages the innate opinions of visited nodes to approximate the equilibrium opinion \(z_{u}\) with an error guarantee of \(|\hat{z}_{u}-z_{u}|\leq\epsilon\). With a running time of \(O(\epsilon^{-2}r^{3}\log r)\), where \(r\) depends on the graph's condition number, the algorithm provides a practical solution for analyzing individual opinion formation in large social networks. 
    
    In practical applications, a fixed walk length and sample size are typically employed to strike a balance between computational efficiency and result accuracy. However, this fixed-parameter approach cannot provide rigorous error guarantees. For instance, as demonstrated in Appendix~\ref{appdix:exp_undi}, the algorithm's performance degrades significantly under power-law distributions of the internal opinions, exhibiting large error fluctuations and poor convergence.
    
\subsection{Algorithm Based on Laplacian Solver}

    The authors in~\cite{xu2021fast} proposed an approximation algorithm based on approximate Laplacian solver~\cite{kyng2016approximate}. By leveraging the positive semi-definiteness of $(\II + \LL)$, the algorithm avoids direct matrix inversion and instead approximates the solution with rigorous error guarantees. For the equilibrium opinion vector, given an error parameter $\epsilon$, the algorithm computes an approximate solution $\hat{\zz}$ in $O\left(m \log^3 n \log\left(\frac{\sqrt{n}}{\epsilon}\right)\right)$ time, satisfying
    \begin{equation*}
        (1 - \epsilon) \|\zz \|_2^2 \le \|\hat{\zz}\|^2_2 \leq (1 + \epsilon) \|\zz \|_2^2.
    \end{equation*}
    
    This algorithm demonstrates strong performance on large-scale undirected graphs, particularly for million-level networks. However, when the number of network nodes is relatively small, its performance is suboptimal due to the fixed overhead of preprocessing steps, such as graph sparsification and recursive decomposition. These steps are necessary for achieving efficiency on large graphs but become disproportionately costly for relatively small graphs. Additionally, due to its space complexity requirements, the algorithm cannot scale to graphs with more than tens of millions of edges. And it is important to note that this algorithm is specifically designed for undirected graphs and does not extend to directed graphs.

\section{Local Iteration Algorithm}\label{sec:localIter}

    In this section, we propose local iterative algorithms, which are local query methods that can efficiently and effectively calculate the equilibrium opinion vector.
    Due to the fact that the derivation regarding correctness and error guarantee can be easily extended to directed graphs, this algorithm is applicable to both undirected and directed graphs.
\subsection{Simple Local Iteration Algorithm}

    In this subsection, we first interpret the equilibrium opinion vector from the perspective of random walks on undirected graphs, and then propose a local iterative method that differs from the Monte Carlo approach. 

\subsubsection{Random Walk Interpretation}
    
    The following lemma illustrates that, the equilibrium opinion vector can be understood not only through its algebraic definition but also more comprehensively by employing a key concept known as the \textit{$\alpha$-discounted random walk}. 

\begin{lemma}\label{lem:walkex}
    For a connected graph $\calG = (V,E)$, the equilibrium opinion vector $\zz$ can be expressed as:
\begin{equation}
    \zz = s_{\mathrm{sum}} \cdot (\II + \DD)^{-1}\sum_{i=0}^\infty\left(\PP^{\top}\DD(\II + \DD)^{-1}\right)^{i}\bar{\sss}, \label{equo}
\end{equation}
    where $\PP = \DD^{-1} \AA$ is the transition matrix of the random walk on $\calG$, and $\bar{\sss}$ is the normalized vector of $\sss$ such that $\sss = s_{\mathrm{sum}} \cdot \bar{\sss}$.
\end{lemma}
\begin{proof}
        For a connected graph $\calG = (V,E)$, we have:
    \begin{align*}
        \zz &= (\II + \LL)^{-1}\sss = (\II + \DD - \AA)^{-1}\sss \\
        &= \left(\left(\II - \AA(\II + \DD)^{-1}\right)(\II + \DD)\right)^{-1}\sss \\
        &= (\II + \DD)^{-1}\left(\II - \AA(\II + \DD)^{-1}\right)^{-1}\sss.
    \end{align*}
    Let $\bar{\sss}$ be a normalized vector of $\sss$, that is to say, $\sss = s_{\mathrm{sum}} \cdot \bar{\sss}$. By the Neumann Series, we expand $\left(\II - \AA(\II + \DD)^{-1}\right)^{-1}$ as follows:
    \begin{align}
        \zz &= (\II + \DD)^{-1}\left(\II - \AA(\II + \DD)^{-1}\right)^{-1}\sss \notag \\
        &= (\II + \DD)^{-1}\sum_{i=0}^\infty\left(\AA(\II + \DD)^{-1}\right)^{i}\sss \notag \\
        &= (\II + \DD)^{-1}\sum_{i=0}^\infty\left(\AA\DD^{-1}\DD(\II + \DD)^{-1}\right)^{i}\sss \notag \\
        &= s_{\mathrm{sum}} \cdot (\II + \DD)^{-1}\sum_{i=0}^\infty\left(\PP^{\top}\DD(\II + \DD)^{-1}\right)^{i}\bar{\sss}, \label{equo-proof}
    \end{align}
    where $\PP = \DD^{-1} \AA$ is the transition matrix of the random walk on $\calG$. This completes the proof.
\end{proof}

     Specifically, a variant of the $\alpha$-discounted random walk can be defined as follows: starting from node $i$ with probability $\bar{s}_i$, at each step where the current node is $j$, the walk either (i) terminates at node $j$ with probability $1/(1+d_j)$, or (ii) transitions uniformly at random to a neighboring node with probability $d_j/(1+d_j)$. The equilibrium opinion vector $\zz$ is obtained through scaling the $\alpha$-discounted random walk's termination distribution with $s_{\mathrm{sum}}$.
    
    However, the heuristic algorithm that directly employs the Monte Carlo method to compute the final result is computationally intensive. As indicated by the Hoeffding Bound~\cite{hoeffding1994probability}, for any given node, a minimum of \(O\left(\frac{1}{\epsilon^2} \log\left(\frac{1}{p_f}\right)\right)\) random walks must be simulated to ensure an absolute error of \(\epsilon\) and a failure probability of \(p_f\). This requirement becomes impractical in scenarios demanding high precision.
    To address this challenge, we explore an alternative perspective to approximate the equilibrium opinion vector.

    Let $\mathbf{r}$ denote the \textit{residual vector}, and define the boundary vector $\mathbf{h} = \epsilon \mathbf{s}$ as the termination criterion for the algorithm, where $\epsilon \in (0,1)$ is an error parameter. The vector $\hat{\zz}$ represents the estimated equilibrium opinion vector. At time $t = 0$, the residual vector is initialized as $\mathbf{r}^{(0)} = \mathbf{s}$, and the estimated equilibrium opinion vector is set to $\hat{\zz}^{(0)} = \mathbf{0}$. At any subsequent time $t$, for any node $i \in V$, if the residue $r^{(t)}_i$ exceeds the corresponding threshold $h_i$, the following updates are performed:
\begin{align*}
    r^{(t+1)}_j &= r^{(t)}_j + \frac{1}{1+d_i} r^{(t)}_i, \quad \forall j \in N_i, \\
    \hat{z}^{(t+1)}_i &= \hat{z}^{(t)}_i + \frac{1}{1+d_i} r^{(t)}_i.
\end{align*}
    In the synchronous case where $\mathbf{h} = \mathbf{0}$, the value $r^{(t)}_i$ corresponds to the probability of a random walk being at node $i$ at time $t$, scaled by $s_{\mathrm{sum}}$. Similarly, $\hat{z}^{(t)}_i$ represents the probability that the random walk terminates at node $i$ before time $(t+1)$, also scaled by $s_{\mathrm{sum}}$. 

\subsubsection{Deterministic Local Iteration Algorithm}

    We propose \textsc{BoundLocalIter}, an algorithm that leverages the local characteristics of nodes to perform asynchronous updates. We use a first-in-first-out (FIFO) queue $Q$ to efficiently manage nodes satisfying the update condition. From the update rule for $r^{(t+1)}_u$, it follows that for any node $v$ and its residual $r_v$, the contribution of node $v$ to the residual of its neighbor $u \in N_v$ is $r_v/(1+d_v)$, while the contribution to the estimated expressed opinion of node $v$ itself is also $r_v/(1+d_v)$. After processing the residual of $v$, $r_v$ is reset to $0$. At each step, \textsc{BoundLocalIter} selects a single node meeting the update condition and performs a push operation. The algorithm terminates when $r_v \leq \epsilon s_v$ for all nodes $v \in V$. The pseudo-code for \textsc{BoundLocalIter} is provided in Algorithm~\ref{alg-localitr}. 

    \begin{algorithm}[t]
        \caption{$\textsc{BoundLocalIter}(\calG, \sss, \epsilon )$}
        \label{alg-localitr}
        \Input{A graph $\calG = (V, E)$, an internal opinion vector $\sss$, a relative error parameter $\epsilon$}
        \Output{Estimated equilibrium opinion vector $\hat{\zz}$}
        {
        \textbf{Initialize} :
            $\hat{\zz} = \zero;\ \rr = \sss$; a FIFO-Queue $Q$\\
        }
        \For{$v \in V$}{
            $Q\text{.append}(v)$\\
        }
        \While{$Q \neq \emptyset$}{
            $v = Q\text{.pop()}$\\
            $\hat{z}_v  = \hat{z}_v  + \frac{r_v}{1+d_v}$\\
            \For{each $u \in N_v$}{
                $r_u = r_u + \frac{r_v}{1+d_v}$\\
                \If{$r_u > \epsilon s_u$ and $u \notin Q$}{ \label{line:con}
                    $Q\text{.append}(u)$\\
                }
            }
            $r_v = 0$
        }
        \textbf{return} $\hat{\zz}$ \\
    \end{algorithm}

    The correctness of the algorithm relies on the relationship between the residual vector and the estimation error, which is guaranteed by the following lemma. 
    \begin{lemma}\label{lem-rela}
        Throughout the execution of Algorithm~\ref{alg-localitr}, the equality $\zz - \hat{\zz} = (\II + \LL)^{-1} \rr$ holds.
    \end{lemma}
    \begin{proof}
        We demonstrate that the invariant holds by using induction. First, we verify that before any computation has begun, the invariant is satisfied by the initialized values:
        \begin{equation*}
            \zz - \hat{\zz}^{(0)} = \zz - \zero = (\II + \LL)^{-1}\rr^{(0)} = (\II + \LL)^{-1}\sss.
        \end{equation*}
        Let $\rr^{(t)}$ denote the residual vector before an update task for any node $i$, and let $\rr^{(t+1)}$ denote the residual vector after the update task. Then a single update task corresponds to the following steps:
        \begin{align*}
            \hat{\zz}^{(t+1)} &= \hat{\zz}^{(t)} + \frac{r^{(t)}_i}{1+d_i}\ee_i,\\
            \rr^{(t+1)} &= \rr^{(t)} -  r^{(t)}_i \ee_i +  r^{(t)}_i\cdot \AA(\II + \DD)^{-1}  \ee_i \\
            &= \rr^{(t)} - r^{(t)}_i (\II - \AA(\II + \DD)^{-1})\ee_i\\
            &= \rr^{(t)} - r^{(t)}_i (\II + \LL)(\II +\DD)^{-1}\ee_i\\
            &= \rr^{(t)} - \frac{r^{(t)}_i}{1+d_i} (\II + \LL)\ee_i.
        \end{align*}
        Assuming that $\zz - \hat{\zz}^{(t)} = (\II+\LL)^{-1}\rr^{(t)}$, it follows that
        \begin{align*}
            \zz - \hat{\zz}^{(t+1)} &= \zz - \hat{\zz}^{(t)} -\frac{r^{(t)}_i}{1+d_i} \ee_i =(\II +\LL)^{-1}\rr^{(t)} - \frac{r^{(t)}_i}{1+d_i}\ee_i\\
            &= (\II +\LL)^{-1}(\rr^{(t+1)} +  \frac{r^{(t)}_i}{1+d_i} (\II + \LL)\ee_i) - \frac{r^{(t)}_i}{1+d_i}\ee_i\\
            &= (\II+\LL)^{-1}\rr^{(t+1)}, 
        \end{align*}
        which completes the proof.
    \end{proof}
    Combining Lemma~\ref{lem-rela}, we can further establish the relative error guarantees for the results returned upon termination of the Algorithm~\ref{alg-localitr}, as shown in the following lemma.
    \begin{lemma}\label{lem-error}
        For any parameter $\epsilon \in (0,1)$, the estimator $\hat{\zz}$ returned by Algorithm~\ref{alg-localitr} satisfies the following relation:
        \begin{equation}
            (1-\epsilon)z_v \le \hat{z}_v \le z_v, \forall v \in V.
        \end{equation}
    \end{lemma}
    \begin{proof}
        On the one hand, since $\rr \ge \zero$, by Lemma~\ref{lem-rela}, we have $z_v \ge \hat{z}_v, \forall v \in V$; on the other hand, the condition in Line~\ref{line:con} of Algorithm~\ref{alg-localitr} indicates that, after termination, the vector $\rr \le \epsilon \sss$, then we have that:
        \begin{equation}
            \zz - \hat{\zz} =(\II +\LL)^{-1} \rr  \le\epsilon (\II +\LL)^{-1} \sss  = \epsilon  \zz.
        \end{equation}
        The equation above implies that $(1-\epsilon)z_v \le \hat{z}_v, \forall v \in V$, hence the estimator $\hat{\zz}$ returned by Algorithm~\ref{alg-localitr} satisfies $(1-\epsilon)z_v \le \hat{z}_v \le z_v, \forall v \in V$.
    \end{proof}

    The complexity of Algorithm~\ref{alg-localitr} is primarily influenced by the cost of performing push operations. According to Line 2, a node $v$ can only be pushed if it satisfies the condition $r_v \ge \epsilon s_v$. Each push operation for node $v$ increases $z_v$ by at least $\epsilon s_v/(1+d_v)$, and this requires $O(d_v)$ time. Thus, the term $\epsilon s_v/(1+d_v)$ can be seen as the \textit{unit gain} of the push operation with respect to the growth of $z_v$.

    Meanwhile, $\hat{z}_v$ is always an underestimate of $z_v$, that is, $\hat{z}_v \le z_v, \forall v \in V$. As a result, for any node $v$, the overall complexity of the push operations is capped by this upper bound $z_v$ divided by the unit gain $\epsilon s_v/(1+d_v)$. Therefore, by considering each node $v \in V$ separately, we derive that total time complexity of $O(\frac{1}{\epsilon} \cdot \sum_{v\in V}\frac{z_v d_v (1+d_v)}{s_v})$ when $s_{\min} > 0$.

    However, this proof approach does not take into account the actual scenario in the algorithm where the residue gradually decreases to a bound, resulting in a relatively loose upper bound on the time complexity. According to Lemma~\ref{lem-rela}, we know that
        $\hat{z}_u = \ee_u^{\top}(\II +\LL)^{-1}(\sss - \rr)$.
    Therefore, the goal of the algorithm is to make each residue smaller, or equivalently, to reduce the sum of the residue vector $r_{\mathrm{sum}}$. In the following, we will provide a tighter time complexity analysis from the perspective of minimizing the sum of the residue vector.

     \begin{theorem}~\label{the:runtime}
        When $s_{\min} > 0$, the upper bound on the running time of Algorithm~\ref{alg-localitr} is 
        \begin{equation*}
            O\left(d_{\max}m \log\frac{ d_{\max} \bar{s}}{\epsilon s_{\min}}\right).
        \end{equation*}
    \end{theorem}
    \begin{proof}
        We start with the case of $\hh = \epsilon s_{\min} \dd$ and prove that when the relative error $\epsilon$ is set to $\epsilon'/d_{\max}$, the bound it produces is less than $\epsilon' s_v$ for any node $v \in V$. This means that the time complexity of Algorithm~\ref{alg-localitr} is less than this case. 

        During the processing, we add a dummy node that does not actually exist. This node is initially placed at the head of the queue and is re-appended to the queue each time it is popped. The set of nodes in the queue when this dummy node is processed for the $(i+1)$-th time is regarded as $S^{(i)}$, and the residue vector at this time is regarded as $\rr^{(i)}$. The sum of the vector $\rr^{(i)}$ is denoted as $r_{\mathrm{sum}}^{(i)}$. The process of handling this set is considered the $(i + 1)$-th iteration. 
        In the context of the $(i+1)$-th iteration, when node $v\in S^{(i)}$ is about to be processed, it holds that $r_v \ge r^{(i)}_v$. Upon the completion of this operation, the sum of expressed vector is increased by $r_v/(1+d_v)$. Consequently, by the conclusion of the $(i+1)$-th iteration, the total reduction in the sum of residue vector amounts to:
        \begin{align}\label{equo:sumprepare}
            r_{\mathrm{sum}}^{(i)} - r_{\mathrm{sum}}^{(i+1)} &= \one^{\top}(\II+\LL)(\hat{\zz}^{(i+1)} - \hat{\zz}^{(i)}) \notag\\
            &= \sum_{v\in S^{(i)}} \frac{r_v}{1+d_v} \ge  \sum_{v\in S^{(i)}}\frac{ r^{(i)}_v}{1+d_v}.
        \end{align}
        Given that the bound for any node $v$ is $\epsilon d_v s_{\min}$, we obtain:
        \begin{equation*}
            \frac{\sum_{v\in S^{(i)}}r^{(i)}_v}{\sum_{v\in S^{(i)}}d_{v}}\geq \epsilon s_{\min}, \text{and} \quad \frac{\sum_{v\notin S^{(i)}}r^{(i)}_v}{\sum_{v\notin S^{(i)}}d_{v}}\leq \epsilon s_{\min}.
        \end{equation*}
        Therefore, it follows that
        \begin{equation*}
            \frac{\sum_{v\in S^{(i)}}r^{(i)}_v}{\sum_{v\in S^{(i)}}d_{v}}\geq \frac{\sum_{v\in S^{(i)}}r^{(i)}_v + \sum_{v\notin S^{(i)}}r^{(i)}_v}{\sum_{v\in S^{(i)}}d_{v} + \sum_{v\notin S^{(i)}}d_{v}} = \frac{r_{\mathrm{sum}}^{(i)}}{m}.
        \end{equation*}
        Substituting the above expression into Eq.~\eqref{equo:sumprepare}, we obtain:
        \begin{align*}
            r^{(i+1)}_{\mathrm{sum}} &\le r_{\mathrm{sum}}^{(i)} - \sum_{v\in S^{(i)}} \frac{r^{(i)}_v}{1+d_v} \le r_{\mathrm{sum}}^{(i)} - \frac{1}{1+d_{\max}} \sum_{v\in S^{(i)}}  r^{(i)}_v \\
            &\le (1-\frac{1}{(1+d_{\max})m}\sum_{v\in S^{(i)}}d_{v})r_{\mathrm{sum}}^{(i)}\\
            &\le \prod_{t=0}^{i}(1-\frac{1}{(1+d_{\max})m}\sum_{v\in S^{(t)}}d_{v}) r_{\mathrm{sum}}^{(0)}.
        \end{align*}
        By utilizing the fact that $1-x\leq e^{-x}$, we obtain that
        \begin{align}\label{equ:timefor}
            r^{(i+1)}_{\mathrm{sum}} &\leq\exp\left(-\sum_{t=0}^{i}(\frac{1}{(1+d_{\max})m} \sum_{v\in S^{(t)}}d_v)\right)r_{\mathrm{sum}}^{(0)}\notag \\
            &=\exp\left(-\frac{1}{(1+d_{\max})m} (\sum_{t=0}^{i}\sum_{v\in S^{(t)}}d_v)\right)\sum_{v\in V} s_v.
        \end{align}

        Recall that the time cost of a single push operation at node $v$ is  $O(d_v)$, let $T^{(i+1)}=\sum_{t=0}^{i}\sum_{v\in S^{(t)}}d_{v}$ be defined as the total time cost prior to the commencement of the $(i+1)$-th iteration. According to Eq.~\eqref{equ:timefor}, to satisfy $r_{\mathrm{sum}}^{(i+1)}\leq \epsilon s_{\min}\sum_{v\in V}d_v = \epsilon m s_{\min}$, it suffices to find the smallest $i$ that meets the following conditions:
        \begin{equation*}
            \exp\left(-\frac{1}{(1+d_{\max})m} T^{(i+1)}\right)\leq \frac{\epsilon m s_{\min}}{n \bar{s}}  \leq\exp\left(-\frac{1}{(1+d_{\max})m} T^{(i)}\right).
        \end{equation*}
        Thus, we obtain $T^{(i)}\leq (1+d_{\max})m\log\frac{n \bar{\sss}}{\epsilon m s_{\min}}\leq T^{(i+1)}$, Given the fact that $T^{(i+1)}-T^{(i)}=\sum_{v\in S^{(i)}}d_{v}\leq m$, we further derive:
        \begin{align*}
            T^{(i+1)}&\leq T^{(i)}+m\leq (1+d_{\max})m\log\frac{n \bar{s}}{\epsilon m s_{\min}}+m \\
            &\leq (1+d_{\max})m\log\frac{\bar{s}}{\epsilon s_{\min}}+m.
        \end{align*}

        For node $v$, the push operation reduces $r_{\mathrm{sum}}$ by $r_v/(1+d_v)$, with a corresponding time cost of $O(d_v)$. Consequently, after incurring a total time cost $T$,  the reduction in $r_{\mathrm{sum}}$ is at least $\epsilon T s_{\min}/(1+d_{\max})$. Hence starting from the state of $r_{\mathrm{sum}} \le  \epsilon m s_{\min}$, the time cost $T$ is bounded by $O((1+d_{\max})m)$, thereby constraining the overall time complexity to $O(d_{\max}m\log\frac{\bar{s}}{\epsilon s_{\min}})$. By setting relative error to $\epsilon/d_{\max}$, the upper bound on the running time of Algorithm~\ref{alg-localitr} becomes $O(d_{\max}m \log\frac{d_{\max} \bar{s} }{\epsilon s_{\min}})$, which completes the proof.
    \end{proof}

\subsubsection{Robustness Enhancement}
    Algorithm~\ref{alg-localitr} offers a technical guarantee of logarithmic relative error; however, its applicability is limited and lacks clarity across diverse scenarios. Specifically, for a node $v$ with an internal opinion of $0$, the iterative loop pertaining to $v$ may fail to terminate. To address this limitation, we enhance Algorithm~\ref{alg-localitr} to introduce Algorithm~\ref{alg-improved}. The refined algorithm incorporates a constant $c>0$ into the internal opinion, thereby ensuring that the iteration bound for any node within the network is strictly positive. This modification effectively prevents excessive iterations and enhances the robustness of the algorithm.

    \begin{algorithm}
        \caption{$\textsc{ImprovedBLI}(\calG, \sss, \epsilon, \sigma, c)$}
        \label{alg-improved}
        \Input{  A graph $\calG = (V, E)$, an internal opinion vector $\sss$, a relative error parameter $\epsilon$, a threshold $\sigma$, a constant $c$}
        \Output{Estimated equilibrium opinion vector $\hat{\zz}$}
        {
        \textbf{Initialize} :
            $\hat{\zz} = \zero;\ \rr = \zero; \epsilon' = \frac{\sigma}{c+\sigma}\epsilon$ \\
        }
        $r_v = s_v + c, \forall v \in V$\\
        $\hat{\zz} = \textsc{BoundLocalIter}(\calG, \sss, \epsilon')$\\
        $\hat{z}_v = \hat{z}_v -c, \forall v \in V$\\
        \textbf{return} $\hat{\zz}$ \\
    \end{algorithm}

    \begin{lemma}
        Given a constant $c > 0$, a relative error parameter $\epsilon$, and a threshold $\sigma$. By setting $\epsilon' = \frac{\sigma}{c+\sigma}\epsilon$, for any equilibrium opinion $z_v > \sigma$, the result $\hat{z}_v$ returned by Algorithm~\ref{alg-improved} satisfies:
        \begin{equation}
            (1-\epsilon)z_v \le \hat{z}_v \le z_v, \forall v \in V.
        \end{equation}
    \end{lemma}
    \begin{proof}
        Let $\sss' = \sss + c\one$, and let $\hat{\zz}' = \hat{\zz} + c\one$, and let $\zz'$ be the equilibrium opinion vector when the internal opinion in the network is $\sss'$. From Lemma~\ref{lem:walkex}, we have $\zz' = (\II+\LL)^{-1} \sss' =(\II+\LL)^{-1}(\sss + c\one) = \zz + c\one$. Therefore, we have:
        \begin{equation*}
            \zz - \hat{\zz} = \zz' - \hat{\zz}'  \le \epsilon' (\II+\LL)^{-1}\sss' = \epsilon'(\zz + c\one) .
        \end{equation*}
        Thus, for any node $v \in V$, we have $z_v - \zz'_v \le (z_v + c)\epsilon'$. Because in a connected graph, as long as there is a node with an internal opinion greater than 0, $z_v$ must be greater than 0. By dividing both sides of the equation by $z_v$ simultaneously, we obtain $\frac{z_v - \hat{z}_v}{z_v} \le \frac{z_v + c}{z_v}\epsilon'$. For any node $v \in V$, such that $z_v \ge \sigma$, we have  $\frac{z_v - \hat{z}_v}{z_v} \le \frac{z_v + c}{z_v}\epsilon' \le (1+\frac{c}{\sigma})\epsilon' = \epsilon$, which completes the proof.
    \end{proof}

    \begin{theorem}
        Let $c = \sigma$, the upper bound on the running time of Algorithm~\ref{alg-improved} is 
        \begin{equation}\label{equ:time1}
            O\left(d_{\max}m\log\frac{d_{\max} \bar{s}}{\epsilon c}\right).
        \end{equation}
    \end{theorem}
    \begin{proof}
        By Theorem~\ref{the:runtime}, when $\epsilon' = \frac{\sigma}{c+\sigma}\epsilon$, we have that the upper bound on the running time of Algorithm~\ref{alg-improved} is
        \begin{equation}
            O\left(d_{\max}m\log{\frac{(1/c+1/\sigma)(\bar{s} +c)d_{\max}}{\epsilon}}\right).
        \end{equation}
        By setting $c = \sigma$, we obtain an upper bound on the running time of $O(d_{\max}m\log\frac{\bar{s}d_{\max}}{\epsilon c})$, which completes the proof.
    \end{proof}

\subsection{Fast Local Iteration Algorithm}
    In each iteration of the \textsc{BoundLocalIter} of Algorithm~\ref{alg-improved}, for the selected node $v$, its estimated expressed opinion $\hat{z}_v$ increases by a unit multiple $r_v$. That is to say, for each given unit quantity $r_v$, the required step size $\omega$ is $1$, which means the increase in expressed opinion in each iteration is $\omega \cdot r_v$. Therefore, an intuitive acceleration method is to increase the step size $\omega$ at each iteration to ensure faster convergence of the expressed opinion. Next, we demonstrate the equivalence between \textsc{BoundLocalIter} and a variant of the Gauss-Seidel method and introduce a faster method based on the SOR technique.

    At the beginning of the $t$-th iteration, let the set of elements in the queue be $S^{(t)}$. The elements are reindexed such that elements from $1$ to $|S^{(t)}|$ belong to $S^{(t)}$ and maintain the same order as in the queue, while elements indexed from $|S^{(t)}|+1$ to $n$ belong to $S^{(t)} \setminus V$. The Gauss-Seidel method updates $\zz$ using the following iteration for solving the linear system $\AA\zz=\bb$, with the residual vector $\rr = \bb - \AA\hat{\zz}$:
    \begin{equation}\label{gauss}
        \hat{z}_i^{(t)}=\frac{1}{a_{ii}}\left[-\sum_{j=1}^{i-1}(a_{ij}\hat{z}_j^{(t)})-\sum_{j=i+1}^n(a_{ij}\hat{z}_j^{(t-1)})+b_i\right], i \in S^{(t)}.
    \end{equation} 
    which computes $\hat{z}_i^{(t)}$ using the most recently calculated values $\hat{z}_{1}^{(t)},\ldots,\hat{z}_{i-1}^{(t)}$. In the standard Gauss-Seidel iteration, the set $S^{(t)} = V$ remains equal to $V$ for every iteration step $t$. 
    The following theorem demonstrates that the \textsc{BoundLocalIter} algorithm is mathematically equivalent to the adapted Gauss-Seidel approach.
    \begin{theorem}\label{the:gauss}
        When $\bb = \sss$ and $\AA = \II + \LL$, the iterations in \textsc{BoundLocalIter} are equivalent to the updates of $\hat{\zz}$ and $\rr$ in the Gauss-Seidel method defined by Eq. ~\eqref{gauss}.
    \end{theorem}
    \begin{proof}
        Assuming at a certain moment $i$ during the $t$-th iteration, we let $\rr_i^{(t)}=(r_{1i}^{(t)},r_{2i}^{(t)},\allowbreak \ldots,r_{ni}^{(t)})^\top$ denote the residual vector for \textsc{BoundLocalIter} corresponding to the approximate solution vector $\hat{\zz}_i^{(t)}$ defined by $\hat{\zz}_i^{(t)}=(\hat{z}_1^{(t)},\hat{z}_2^{(t)},\allowbreak\ldots,\hat{z}_{i-1}^{(t)},\hat{z}_i^{(t-1)},\allowbreak \ldots,\hat{z}_n^{(t-1)})^\top$. According to Lemma~\ref{lem-rela}, we have:
        \begin{equation*}
            \rr^{(t)}_i = (\II + \LL)(\zz - \hat{\zz_i}) = \sss - (\II+\LL)\hat{\zz_i},
        \end{equation*}
        hence, let $\bb = \sss$ and $\AA = \II + \LL$, we demonstrate the $k$-th component of $\rr_i^{(t)}$ is 
        \begin{equation}\label{equ:guass_trans1}
            r_{ki}^{(t)}=b_k-\sum_{j=1}^{i-1}a_{kj}\hat{z}_j^{(t)}-\sum_{j=i}^na_{kj}\hat{z}_j^{(t-1)}.
        \end{equation}
        Specifically, when $k = i$, we have:
        \begin{equation*}
            r_{ii}^{(t)}=b_i-\sum_{j=1}^{i-1}a_{ij}\hat{z}_j^{(t)}-\sum_{j=i+1}^na_{ij}\hat{z}_j^{(t-1)} - a_{ii}\hat{z}_i^{(t-1)}.
        \end{equation*}
        By moving $\hat{z}_i^{(t-1)}$ to the left-hand side of the equation, we obtain:
        \begin{equation}\label{equ:guass_trans2}
             a_{ii}\hat{z}_i^{(t-1)} + r_{ii}^{(t)}=b_i-\sum_{j=1}^{i-1}a_{ij}\hat{z}_j^{(t)}-\sum_{j=i+1}^na_{ij}\hat{z}_j^{(t-1)}.
        \end{equation}
        Recall Algorithm~\ref{alg-localitr}, where $\hat{z}_i^{(t)} = \hat{z}_i^{(t-1)} + \frac{r_{ii}^{(t)}}{1+d_i}$, where $1/(1+d_i) = 1/a_{ii}$. Consequently, we ultimately obtain:
        \begin{equation*}
            \hat{z}_i^{(t)} = \hat{z}_i^{(t-1)} + \frac{1}{a_{ii}} r_{ii}^{(t)} = \frac{1}{a_{ii}}\left[b_i-\sum_{j=1}^{i-1}a_{ij}\hat{z}_j^{(t)}-\sum_{j=i+1}^na_{ij}\hat{z}_j^{(t-1)}\right],
        \end{equation*}
        which exhibits the same form as Eq.~\eqref{gauss}, thus completing the proof.
    \end{proof}

Theorem~\ref{the:gauss} implies that \textsc{BoundLocalIter} can be accelerated using techniques originally developed for the Gauss-Seidel method. A well-established approach for enhancing Gauss-Seidel convergence is the SOR method \cite{young1954iterative}. This technique incorporates a relaxation parameter $\omega$ to modulate update magnitudes between iterations. Optimal selection of $\omega$ yields substantially faster convergence relative to conventional Gauss-Seidel implementations. Applying this acceleration strategy to \textsc{BoundLocalIter}, the SOR-modified update for $\hat{z}^{(t)}$ becomes:
    \begin{equation*}
        \hat{z}_i^{(t)}=(1-\omega)\hat{z}_i^{(t-1)}+\frac{\omega}{a_{ii}}\left[b_i-\sum_{j=1}^{i-1}a_{ij}\hat{z}_j^{(t)}-\sum_{j=i+1}^na_{ij}\hat{z}_j^{(t-1)}\right],
    \end{equation*}
    where $\omega \in (0,2)$. By employing Eq.~\eqref{equ:guass_trans2}, the update of the estimator for the corresponding expressed opinion of node $i$ becomes:
    \begin{align}
        \hat{z}_i^{(t)}&=(1-\omega)\hat{z}_i^{(t-1)} +\frac{\omega}{a_{ii}} (a_{ii}\hat{z}_i^{(t-1)} + r_{ii}^{(t)})\notag\\
        &=\hat{z}_i^{(t-1)} + \frac{\omega}{a_{ii}} r_{ii}^{(t)} =\hat{z}_i^{(t-1)} + \frac{\omega}{1+d_i}\cdot r_{ii}^{(t)}. \label{equ:sor1}
    \end{align}
    Meanwhile, by utilizing Eq.~\eqref{equ:guass_trans1}, the update of the corresponding residual values becomes:
    \begin{align*}
        r_{k(i+1)}^{(t)} &=  b_k-\sum_{j=1}^{i}a_{kj}\hat{z}_j^{(t)}-\sum_{j=i+1}^na_{kj}\hat{z}_j^{(t-1)}\\
        & = b_k-\sum_{j=1}^{i-1}a_{kj}\hat{z}_j^{(t)}-\sum_{j=i}^na_{kj}\hat{z}_j^{(t-1)} - a_{ki}\hat{z}_i^{(t)} + a_{ki}\hat{z}_i^{(t-1)} \\
        & = r_{ki}^{(t)} -  \omega\cdot \frac{a_{ki}}{a_{ii}} r_{ii}^{(t)}.
    \end{align*}
    Note that when node $k$ does not belong to the inner neighbors of node $i$, the corresponding value $a_{ki}=0$. Conversely, if node $k$ belongs to the inner neighbors of $i$, the corresponding value $a_{ki}=-1$. Since $a_{ii} = 1+d_i$, we ultimately obtain:
    \begin{align}
        r_{i(i+1)}^{(t)} &= (1-  \omega)\cdot r_{ii}^{(t)},\label{equ:sor2}\\
        r_{k(i+1)}^{(t)} &= r_{ki}^{(t)} + \frac{\omega}{1+d_i}\cdot r_{ii}^{(t)},\quad \forall k \in N_i.\label{equ:sor3}
    \end{align}

    \begin{algorithm}[t]
        \caption{$\textsc{BoundLocalIterSOR}(\calG, \sss, \epsilon, \omega )$}
        \label{alg-sor}
        \Input{  A graph $\calG = (V, E)$, an internal opinion vector $\sss$, a relative error parameter $\epsilon$, a constant $\omega$}
        \Output{Estimated equilibrium opinion vector $\hat{\zz}$}
        {
        \textbf{Initialize} :
            $\hat{\zz} = \zero;\ \rr = \sss$; a FIFO-Queue $Q$\\
        }
        \For{ $v \in V$}{
            $Q\text{.append}(v)$\\
        }
        \While{$Q \neq \emptyset$}{
            $v = Q\text{.pop}()$\\
            $\hat{z}_v = \hat{z}_v + \frac{\omega}{1+d_v}r_v$\\
            \For{each $u \in N_v$}{
                $r_u = r_u +\frac{\omega}{1+d_v} r_v$\\
                \If{$|r_u| > \epsilon s_u$ and $u \notin Q$}{\label{line:1con}
                    $Q\text{.append}(u)$\\
                }
            }
            $r_v = (1-\omega) \cdot r_v$\\
            \If{$|r_v| > \epsilon s_v$ and $v \notin Q$}{\label{line:2con}
                $Q\text{.append}(v)$\\
            }
        }
        \textbf{return} $\hat{\zz}$ \\
    \end{algorithm}
    
     Building upon the over-relaxed formulations in Eqs. ~\eqref{equ:sor1}, \eqref{equ:sor2} and \eqref{equ:sor3}, we develop \textsc{BoundLocalIterSOR} , as detailed in Algorithm \ref{alg-sor}, which maintains implementation simplicity by requiring only the relaxation parameter $\omega$ as input. Notably, this generalized algorithm reduces to standard \textsc{BoundLocalIter} when $\omega = 1$. By removing the non-negativity constraint on elements of $r_v$, our approach achieves both computational acceleration and preservation of relative error guarantees, as formalized below:

    \begin{lemma}\label{lem-sor_error}
        For any parameter $\epsilon \in (0,1)$, the estimator $\hat{\zz}$ returned by Algorithm~\ref{alg-sor} satisfies the following relation:
        \begin{equation}
            (1-\epsilon)z_v \le \hat{z}_v \le (1+\epsilon)z_v, \forall v \in V.
        \end{equation}
    \end{lemma}
    \begin{proof}
       The condition specified in Lines~\ref{line:1con} and~\ref{line:2con} of Algorithm~\ref{alg-sor} demonstrates that upon termination, the residual satisfies $|\mathbf{r}| \leq \epsilon \mathbf{s}$. Given that matrix $(\II +\LL)^{-1}$ is positive matrix, we consequently derive the following:
        \begin{equation}
            |\zz - \hat{\zz}| =|(\II+\LL)^{-1} \rr|  \le\epsilon (\II+\LL)^{-1} \sss  = \epsilon  \zz.
        \end{equation}
        The equation above implies that $(1-\epsilon)z_v \leq \hat{z}_v \leq (1+\epsilon)z_v$ for all $v \in V$, as desired.
    \end{proof}

    Similar to the limitations observed in \textsc{BoundLocalIter}, \textsc{BoundLocalIterSOR} may also encounter challenges in ensuring termination under certain conditions, particularly for nodes with internal opinion close to zero. To address this issue, we propose an enhanced version of the algorithm, referred to as \textsc{ImprovedBLISOR}, which incorporates the same methodological improvement introduced in Algorithm~\ref{alg-improved} and maintains identical error guarantees.

    \textbf{Parameter Choosing for $\omega$.} 
    It is worth noting that the matrix $\II+\LL$ is symmetric and positive-definite for undirected graphs, ensuring that SOR converges for any $\omega \in (0, 2)$. Let $\mu = \rho((\II+\DD)^{-1}\AA)$, where $\rho(\AA)$ denotes the spectral radius of matrix $\AA$, the optimal value of $\omega$~\cite{hackbusch1994iterative} is given by
    \begin{equation}
    \omega_{\text{opt}} = 1+\left(\frac{\mu}{1+\sqrt{1-\mu^2}}\right)^2. \label{equ:optimal-omega}
    \end{equation}
    In practical networks, the value of $\mu$ is often difficult to compute. Nevertheless, a heuristic approach can be employed to determine $\omega$: Begin with \(\omega \to 2^-\) and iteratively reduce \(\omega\) by a fixed step size until it reaches \(1\). The value of \(\omega\) that achieves the \emph{minimum number of updates} is then determined. In practical scenarios, the parameter \(\omega\) in undirected graphs is typically set to 1.5, which results in an acceleration of more than approximately three times compared to the \textsc{BoundLocalIter} method.

\section{Experiments}\label{sec:experim}

    In this section, we conduct experiments on 8 real-world benchmark graphs to evaluate our proposed algorithms, all of which are undirected graphs. For directed graphs, relevant experiments are included in Appendix~\ref{appdix:exp_di}. For the sake of convenience, in the following text, we will abbreviate the algorithm \textsc{ImprovedBLI} as \textsc{BLI}, the algorithm \textsc{ImprovedBLISOR} as \textsc{BLISOR}, and the algorithm \textsc{Laplacian Solver} as \textsc{Solver}.

    \textbf{Machine.} Our extensive  experiments  were conducted on a machine equipped with 12-core 1.7GHz  Intel i5-1240P CPU and 40GB of main memory. The code for all our algorithms is implemented using \textit{Julia v1.10.7}.  The  solver in \textsc{Solver} we use is based on the technique in~\cite{kyng2016approximate}.

\begin{table*}[t!]
    \caption{Statistics of real-world networks used in our experiments and comparison of running time (seconds) between \textsc{RWB}, \textsc{Solver}, \textsc{BLI}, and \textsc{BLISOR}. \textsc{RWB} and \textsc{Solver} are evaluated only under uniform distribution (Unif), while \textsc{BLI} and \textsc{BLISOR} are tested under Unif, Exp (exponential distribution), and Pareto (power-law distribution).} 
    \label{tab:running-times_ungraph}
    \centering
    \begin{adjustbox}{max width=\textwidth}
    \begin{tabular}{cccc cc ccc ccc c}
    \toprule
        \multicolumn{4}{c}{Dataset}
        & \multicolumn{8}{c}{Time}
        & \multirow{3}{*}{$\omega$}\\
    \cmidrule(lr){1-4}
    \cmidrule(lr){5-12}
    
    \multirow{2}{*}{name}&\multirow{2}{*}{$n$}&\multirow{2}{*}{$m$}&\multirow{2}{*}{$d_{\max}$}
    & \multirow{2}{*}{{\textsc{RWB}}} 
    & \multirow{2}{*}{{\textsc{Solver}}} 
    &  \multicolumn{3}{c}{{\textsc{BLI}}}
    & \multicolumn{3}{c}{{\textsc{BLISOR}}} & \\ 
    \cmidrule(lr){7-9}
    \cmidrule(lr){10-12}
    &&&&&& Unif & Exp & Pareto
    & Unif & Exp & Pareto &\\
    \midrule
    Hamsterster & 2,426&16,630&273 
    & 2.74 & 4.02 
    & 7e-3 & 8e-3 & 8e-3
    & 3e-3 & 3e-3 & 3e-3
    &1.5\\
    Email-Enron & 33,696& 180,811 & 1,383
    & 35.67 & 4.61
    & 0.09 & 0.08 & 0.11 
    & 0.02 & 0.02 & 0.02
    &1.6 \\
    DBLP & 317,080& 1,049,866 &343 
    & 283.80 &  5.80 
    & 0.84 & 0.80 & 0.88
    & 0.18 & 0.18 & 0.22
    &1.5 \\
    YoutubeSnap &1,134,890& 2,987,624& 28,754
    & 1912.85 & 8.66 
    & 3.07 & 2.88 & 3.28 
    & 0.68 & 0.69 & 0.85
    & 1.5 \\
    Flixster &2,523,386& 7,918,801& 1,474
    & 6964.04 & 20.81 
    & 7.54 & 7.16 & 7.10
    & 2.07 & 2.00 & 2.19
    & 1.55\\
    Orkut &3,072,441& 117,185,083& 33,313
    & 170815.42 & - 
    & 990.18 & 1065.41 & 1082.50
    & 206.34 & 221.11 & 218.51
    &1.65\\
    LiveJournal &4,033,137& 27,933,062& 2,651 
    & 28048.85 & 76.32 
    & 69.35 & 70.39 & 70.21
    & 18.54 & 19.16 & 18.48
    &1.5 \\
    SinaWeibo &58,655,849& 261,321,071& 278,491
    & 449274.47 &  - 
    & 924.45 & 957.79 & 921.55
    & 285.8  & 300.66 & 253.25
    &1.5 \\
    \bottomrule
    \end{tabular}
    
    \end{adjustbox}
\end{table*}

\begin{table*}[t!]
    \caption{Relative error for estimated $\zz,\Controversy,\Disagreement,\Internal,\Polarization$ under uniform distribution, comparing \textsc{RWB}, \textsc{BLI}, and \textsc{BLISOR}.}
    \label{tab:measure_li}
    \centering
    \begin{adjustbox}{max width=\textwidth}
    \begin{tabular}{c ccccc ccccc ccccc}
    \toprule
        \multirow{3}{*}{Dataset}
        &\multicolumn{15}{c}{Relative error of equilibrium opinion and four estimated quantities for three algorithms in uniform distribution($\times 10^{-2}$)}\\
    \cmidrule(lr){2-16}
    & \multicolumn{5}{c}{\textsc{RWB}} & \multicolumn{5}{c}{BLI}&\multicolumn{5}{c}{BLISOR}\\
    \cmidrule(lr){2-6}
    \cmidrule(lr){7-11}
    \cmidrule(lr){12-16}
    & $\zz$ &$\Controversy$&$\Disagreement$&$\Internal$&$\Polarization$
    & $\zz$ &$\Controversy$&$\Disagreement$&$\Internal$&$\Polarization$
    & $\zz$ &$\Controversy$&$\Disagreement$&$\Internal$&$\Polarization$\\
    \midrule
    Hamsterster  & 7.2520 & 0.0477 & 0.1963 & 0.3877 & 1.0447
                 & 0.8676 & 0.9284 & 1.0806 & 0.1820 & 0.8627
                 & 0.7536 & 0.2793 & 0.3741 & 0.0618 & 0.0856\\
    Email-Enron  & 6.7352 & 0.1199 & 7.8120 & 0.1359 & 0.6580
                 & 0.8128 & 0.9482 & 0.9766 & 0.3076 & 0.9751
                 & 0.6381 & 0.1194 & 0.1995 & 0.0546 & 0.2218\\
    DBLP         & 7.0512 & 0.4913 & 6.0803 & 0.0178 & 1.1419
                 & 0.8621 & 0.9541 & 1.0092 & 0.2978 & 0.9781
                 & 0.7472 & 0.1274 & 0.0936 & 0.0296 & 0.0451\\
    YoutubeSnap  & 8.1089 & 0.0782 & 0.6488 & 2.4370 & 2.0180
                 & 0.8661 & 0.9530 & 1.0027 & 0.5167 & 0.9970
                 & 0.7665 & 0.5077 & 0.1214 & 0.0520 & 0.1128\\
    Flixster     & 7.8437 & 0.1965 & 4.5266 & 0.7881 & 0.7838
                 & 0.8510 & 0.9566 & 1.0049 & 0.5564 & 1.0064
                 & 0.7380 & 0.0610 & 0.0611 & 0.0424 & 0.1267\\
    LiveJournal  & 14.8626 & 0.0836 & 9.4652 & 0.7887 & 0.0804
                 & 0.8773 & 0.9605 & 1.0290 & 0.2445 & 1.0527
                 & 0.8066 & 0.3117 & 0.2608 & 0.0617 & 0.1671\\
    \bottomrule
    \end{tabular}
    
    \end{adjustbox}
\end{table*}

    \textbf{Datasets.} All the real-world networks we consider are  publicly available in the Koblenz Network Collection~\cite{kunegis2013konect}, SNAP~\cite{leskovec2016snap} and Network Repository~\cite{rossi2015network}. The first four columns of Table~\ref{tab:running-times_ungraph} are related information of networks, including  the network name, the number of nodes, the number of edges, and the maximum node degree. The smallest network consists of 2426 nodes, while the largest network  has more than ten million. In Table~\ref{tab:running-times_ungraph}, the networks listed in an increasing order of the number of nodes.

    \textbf{Internal Opinion Distributions.} In our experimental framework, we generate internal node opinions using three distinct probability distributions: uniform, exponential, and power-law. These distributions are selected to model fundamentally different scenarios of opinion formation in networked systems. For the uniform distribution, each node's opinion value is independently drawn from a uniform distribution over the interval $[0,1]$. The exponential distribution implementation involves generating positive opinion values according to the probability density function $f(x) = \mathrm{e}^{x_{\min}} \mathrm{e}^{-x}$ with  $x_{\min} > 0$, which we then normalize to the unit interval by dividing each value by the maximum observed sample. For modeling scenarios where opinions follow a heavy-tailed distribution, we employ a power-law distribution with exponent $\alpha = 2.5$, implemented through the probability density function $f(x) = (\alpha-1)x_{\min}^{\alpha-1}x^{-\alpha}$. Similar to the exponential case, the generated values undergo normalization to the $[0,1]$ range, preserving the distribution's characteristic property where a small number of nodes hold significantly higher opinions than the majority.

    \textbf{Baseline.} We compare against the algorithm \textsc{RWB} and algorithm \textsc{Solver}. For the \textsc{RWB} algorithm, we conducted tests on all nodes, and as mentioned in~\cite{neumann2024sublinear}, the RWB algorithm performs well in single node or small batch computing. Similar to ~\cite{neumann2024sublinear}, we set the step size to 600 and the number of simulations to 4000. For the algorithm \textsc{Solver}, We run their algorithm with $\epsilon =10^{-6}$. We do not compare against an exact baseline, since the experiments in~\cite{xu2021fast} show that their algorithm has a negligible error in practice and since the exact computation is infeasible for our large datasets (in the experiments of~\cite{xu2021fast}, their algorithm’s relative error is less than $10^{-6}$ and matrix inversion does not scale to graphs with more than $56,000$ nodes).

    \textbf{Experiment Settings.} Eq.~\eqref{equ:time1} represents a function of time complexity with a constant $c > 0$. Theoretically, we can analyze this function to determine its minimum value over the positive interval. In practice, we set $c=10^{-5}$ and simultaneously set $\sigma = 10^{-5}$. For the parameter $\omega$, we conducted tests within the interval $[1, 2)$ with a step size of $0.5$, aiming to select the optimal $\omega$ as precisely as possible. It should be noted that, as shown in Figure~\ref{fig:wpara} of Appendix~\ref{appdix:exp_undi}, for undirected graphs, setting $\omega=1.5$ can achieve approximately optimal results.

\begin{figure}[t]
    \centering
    \includegraphics[width=\linewidth]{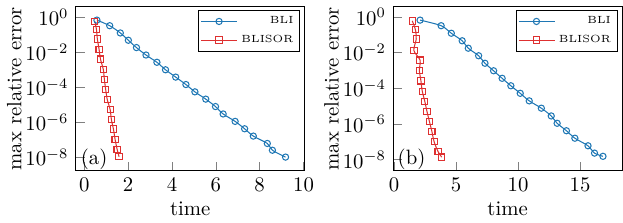}
    \caption{Actual relative-error v.s. execution time (seconds) on two real-world datasets: (a) YoutubeSnap, (b) Flixster. }
    \label{fig:timeerror}
\end{figure} 

    \textbf{Efficiency.} Table~\ref{tab:running-times_ungraph} reports the running times of algorithms \textsc{BLI}, \textsc{BLISOR}, \textsc{RWB} and \textsc{Solver} on various networks used in our experiments. We note that due to the high memory and time costs, we were unable to compute the relevant quantities for the Orkut and SinaWeibo networks using the \textsc{Solver} algorithm, but we were able to run the \textsc{BLI} and \textsc{BLISOR} algorithms. Since both \textsc{BLI} and \textsc{BLISOR} provide relative error guarantees, we set the error parameter $\epsilon = 10^{-2}$ for most of our experiments, which suffices for most practical scenarios. For algorithms \textsc{RWB} and \textsc{Solver}, we only present the running times under the uniform distribution, as their operational efficiency is almost unaffected by the internal opinions. For the three internal opinion distributions across different networks, we recorded the running times of the algorithms \textsc{BLI} and \textsc{BLISOR}. Table~\ref{tab:running-times_ungraph} shows that for all considered networks, the running time of \textsc{BLI} is less than that of \textsc{RWB} and \textsc{Solver}, while \textsc{BLISOR} is approximately three times faster than \textsc{BLI} on almost all networks. For all networks with fewer than ten million edges, \textsc{BLI} demonstrates a significant improvement in running time compared to \textsc{Solver}. For instance, on the Email-Enron network, the running time of \textsc{BLI} is less than one percent of that of \textsc{Solver}. Finally, for extremely large graphs, particularly those with hundreds of millions of edges, our proposed algorithms exhibit a substantial efficiency advantage, as \textsc{Solver} cannot be executed on these networks. Figure~\ref{fig:timeerror} illustrates the relationship between the true relative error and time for algorithms \textsc{BLI} and \textsc{BLISOR} on the YoutubeSnap and Flixster networks, demonstrating that the maximum relative error of our algorithms decays exponentially over time, with \textsc{BLISOR} converging faster than \textsc{BLI}.

    \textbf{Accuracy.} Besides efficiency, our algorithms \textsc{BLI} and \textsc{BLISOR} also achieve high precision. We evaluate their accuracy under uniform distribution in Table~\ref{tab:measure_li}, using \textsc{RWB} as the baseline. Results for exponential and power-law distributions are in Appendix~\ref{appdix:exp_undi}. For each internal opinion distribution, we compare approximate results from our algorithms with exact solutions. For Hamsterster and Email-Enron networks, exact solutions are derived via matrix inversion; for others, we use \textsc{Solver}'s results (precision $10^{-6}$~\cite{xu2021fast}), except for Orkut and SinaWeibo in Table~\ref{tab:running-times_ungraph}. For each quantity $\rho$, we compute the relative error $\sigma = |\rho - \tilde{\rho}| / \rho$. Table~\ref{tab:measure_li} shows the maximum relative error of equilibrium opinion vector $\zz$ and errors for four estimators: $\Controversy$, $\Disagreement$, $\Internal$, and $\Polarization$. Results show that when $\epsilon=10^{-2}$, all four estimators exhibit highly consistent relative errors, all below $0.011$. Moreover, \textsc{BLISOR} generally outperforms \textsc{BLI}, and the maximum relative error of equilibrium opinions is constrained below $10^{-2}$, surpassing \textsc{RWB}.

\section{Conclusion}\label{sec:conclusion}
    In this paper, we tackled the challenge of efficiently computing the equilibrium opinion vector and related metrics in both directed and undirected social networks under the Friedkin-Johnsen model. Traditional methods like matrix inversion are hindered by high computational or spatial complexity, making them unsuitable for large-scale networks. To address this, we introduced a novel local iteration-based algorithm that approximates the equilibrium opinion vector with guaranteed relative error bounds. Our algorithm uses asynchronous push operations to iteratively refine the solution, significantly improving efficiency and scalability compared to methods like the Laplacian solver. We enhanced robustness by adding a constant $c$ to ensure convergence when nodes have zero internal opinions and optimized it with Successive Over-Relaxation techniques to speed up convergence. Extensive experiments on real-world networks showed that our algorithm is highly efficient and scalable, especially for large-scale networks with tens of millions of nodes.

\begin{acks}
The work was supported by  the National Natural Science Foundation of China (No. 62372112). G.Y. Wang was also supported by Fudan's Undergraduate Research Opportunities Program, FDUROP under Grant No. 23944.  
\end{acks}

\newpage
\bibliographystyle{ACM-Reference-Format}
\bibliography{main}

\appendix
\balance
    
    \begin{algorithm}
        \caption{$\textsc{ForestSample}(\mathcal{G},\sss, l)$}
        \label{alg:sampleforest}
        \Input{
            a digraph $\mathcal{G}$, a internal opinion vector $\sss$, number of generated spanning forests $l$
        }
        \Output{Estimated equilibrium opinion vector $\hat{\zz}$}
        
        {
        \textbf{Initialize} :
            $\hat{\zz} \leftarrow \zero$ \\
        }
        
        \For{$t = 1$ \KwTo $l$}{
            RootIndex $\leftarrow$ \textsc{RandomForest}($\mathcal{G}$)\\
            \For{$i = 1$ \KwTo $n$}{
                $u \leftarrow$ RootIndex[$i$]\\
                $\hat{z}_i = \hat{z}_i + s_u$\\
            }
        }
        $\hat{z}_u \leftarrow \hat{z}_u/l,\forall u \in V$\\

        \Return $\hat{\zz}$\\
    \end{algorithm}
\section{Forest Sampling in Directed Graphs}\label{appdix:direct}

    In a directed graph, due to the fact that $\II+\LL$ dose not satisfy the property of positive semi-definiteness, we cannot use the Laplacian solver to solve the problem. The authors in ~\cite{sunOpinionOptimizationDirected2023} introduce a heuristic approach to efficiently compute the average expressed opinion in directed social networks, particularly relevant for scenarios aiming to influence public opinion. By extending Wilson's algorithm for generating uniform spanning trees, the authors devise a method to sample spanning converging forests. This technique allows for the estimation of the average expressed opinion without the need to directly compute the inverse of the forest matrix. The sampling method not only alleviates computational burden but also enhances scalability to massive graphs comprising over twenty million nodes. As a variant of the original algorithm, algorithm \textsc{ForestSample}, detailed in Algorithm ~\ref{alg:sampleforest}, can solve the expressed opinion vector, which has a linear time complexity $O (ln)$ relative to the sampling number $l$, thereby providing a practical and efficient solution for analyzing large-scale social networks.

\section{Additional Experiments}

\subsection{Additional Experiments in Undirected Graphs}\label{appdix:exp_undi}
    \textbf{Experiment Settings.} In this study, we investigate the impact of different values of the relaxation factor $\omega$ on the number of updates performed by the \textsc{BLISOR} algorithm, and the results are presented in Figure~\ref{fig:wpara}. The experiments are conducted on four real-world datasets: (a) Email-Enron, (b) DBLP, (c) YoutubeSnap, and (d) Flixster. Our results demonstrate that the choice of $\omega$ significantly influences the convergence behavior of \textsc{BLISOR}. Specifically, for undirected graphs, the optimal performance is achieved when $\omega$ is around 1.5. At this value, the number of updates required for convergence is minimized, indicating a balance between computational efficiency and algorithmic stability.

\begin{table*}[t]
    \caption{Relative error for estimated $\zz, \Controversy,\Disagreement,\Internal,\Polarization$ under exponential distribution, comparing \textsc{RWB}, \textsc{BLI}, and \textsc{BLISOR}.}.
    \centering
    \begin{adjustbox}{max width=\textwidth}
    \label{tab:measure1}
    \begin{tabular}{c ccccc ccccc ccccc}
    \toprule
        \multirow{3}{*}{Dataset}
        &\multicolumn{15}{c}{Relative error of equilibrium opinion and four estimated quantities for three algorithms in exponential distribution($\times 10^{-2}$)}\\
    \cmidrule(lr){2-16}
    & \multicolumn{5}{c}{\textsc{RWB}} & \multicolumn{5}{c}{BLI}&\multicolumn{5}{c}{BLISOR}\\
    \cmidrule(lr){2-6}
    \cmidrule(lr){7-11}
    \cmidrule(lr){12-16}
    & $\zz$ &$\Controversy$&$\Disagreement$&$\Internal$&$\Polarization$
    & $\zz$ &$\Controversy$&$\Disagreement$&$\Internal$&$\Polarization$
    & $\zz$ &$\Controversy$&$\Disagreement$&$\Internal$&$\Polarization$\\
    \midrule
    Hamsterster  & 7.0152 & 0.2671 & 0.2550 & 0.1778 & 0.1731
                 & 0.8132 & 0.9480 & 1.0459 & 0.1639 & 0.9838
                 & 0.6241 & 0.2977 & 0.3390 & 0.0528 & 0.1896\\
    Email-Enron  & 8.8191 & 0.7693 & 3.7051 & 1.6334 & 2.6014
                 & 0.8376 & 0.9827 & 1.0561 & 0.3269 & 1.0798
                 & 0.8098 & 0.1691 & 0.3407 & 0.1003 & 0.4291\\
    DBLP         & 7.5720 & 1.7260 & 10.8969 & 3.4766 & 7.9725
                 & 0.8665 & 0.9669 & 1.0092 & 0.2893 & 0.9862
                 & 0.7795 & 0.1145 & 0.0730 & 0.0220 & 0.0286\\
    YoutubeSnap  & 9.4665 & 0.8511 & 0.9242 & 0.1219 & 3.8246
                 & 0.8841 & 0.9437 & 0.9610 & 0.4827 & 0.9461
                 & 0.8226 & 0.1491 & 0.3558 & 0.1710 & 0.3486\\
    Flixster     & 8.3171 & 0.6716 & 2.1735 & 0.5375 & 1.2377
                 & 0.8664 & 0.9584 & 0.9910 & 0.5353 & 0.9790
                 & 0.8096 & 0.1980 & 0.6314 & 0.3413 & 0.7434\\
    LiveJournal  & 11.8273 & 0.1405 & 1.1572 & 1.0840 & 2.8327 
                 & 0.8853 & 0.9668 & 1.0179 & 0.2338 & 1.0401
                 & 0.8379 & 0.3051 & 0.2509 & 0.0577 & 0.1723\\
    \bottomrule
    \end{tabular}
    
    \end{adjustbox}
\end{table*}

\begin{table*}[t] 
    \caption{Relative error for estimated $\zz, \Controversy,\Disagreement,\Internal,\Polarization$ under power-law distribution, comparing \textsc{RWB}, \textsc{BLI}, and \textsc{BLISOR}.}.
    \centering
    \begin{adjustbox}{max width=\textwidth}
    \label{tab:measure2}
    \begin{tabular}{c ccccc ccccc ccccc}
    \toprule
        \multirow{3}{*}{Dataset}
        &\multicolumn{15}{c}{Relative error of equilibrium opinion and four estimated quantities for three algorithms in power-law distribution($\times 10^{-2}$)}\\
    \cmidrule(lr){2-16}
    & \multicolumn{5}{c}{\textsc{RWB}} & \multicolumn{5}{c}{BLI}&\multicolumn{5}{c}{BLISOR}\\
    \cmidrule(lr){2-6}
    \cmidrule(lr){7-11}
    \cmidrule(lr){12-16}
    & $\zz$ &$\Controversy$&$\Disagreement$&$\Internal$&$\Polarization$
    & $\zz$ &$\Controversy$&$\Disagreement$&$\Internal$&$\Polarization$
    & $\zz$ &$\Controversy$&$\Disagreement$&$\Internal$&$\Polarization$\\
    \midrule
    Hamsterster  & 24.2161 & 4.0837 & 8.1190 & 7.4877 & 6.5275
                 & 0.8396 & 1.2290 & 1.5415 & 0.3793 & 1.3705
                 & 0.7682 & 0.7900 & 1.3841 & 0.3403 & 1.0176\\
    Email-Enron  & 45.7143 & 30.0040 & 33.2984 & 25.2325 & 37.8610
                 & 0.8901 & 0.9309 & 1.0011 & 0.3476 & 0.9142
                 & 0.8685 & 0.4781 & 0.3614 & 0.1265 & 0.6310\\
    DBLP         & 57.0701 & 44.2895 & 44.1846 & 33.7844 & 51.5504
                 & 0.8995 & 1.0322 & 0.9264 & 0.2059 & 1.0391
                 & 0.8144 & 0.1964 & 0.1506 & 0.0329 & 0.2449\\
    YoutubeSnap  & 345.8933 & 91.7873 & 98.7708 & 96.4878 & 93.4981
                 & 0.9073 & 0.7815 & 0.7647 & 0.4100 & 0.7781
                 & 0.8729 & 0.4456 & 0.5200 & 0.2763 & 0.4497\\
    Flixster     & 1359.7915 & 92.5764 & 97.3794 & 98.7343 & 93.3010
                 & 0.9071 & 0.5682 & 0.5644 & 0.1903 & 0.5650
                 & 0.8786 & 1.6239 & 1.6370 & 0.5541 & 1.6329\\
    LiveJournal  & 131.0797 & 79.9158 & 85.1028 & 87.1971 & 85.0056 
                 & 0.9074 & 0.8705 & 0.9397 & 0.2287 & 0.8640
                 & 0.8989 & 0.3121 & 0.3707 & 0.0904 & 0.3108\\
    \bottomrule
    \end{tabular}
    
    \end{adjustbox}
\end{table*}

    \textbf{Accuracy.} We evaluate the relative error for estimated quantities $\zz$, $\Controversy$, $\Disagreement$, $\Internal$, and $\Polarization$ under both exponential and power-law distributions, comparing the performance of \textsc{RWB}, \textsc{BLI}, and \textsc{BLISOR}. The results, presented in Table~\ref{tab:measure1} and Table~\ref{tab:measure2}, demonstrate that \textsc{BLI} and \textsc{BLISOR} achieve highly accurate and stable performance, with all relative errors below $0.02$. This indicates the effectiveness and robustness of our algorithms in estimating the relevant quantities across different distributions. In contrast, \textsc{RWB} exhibits inconsistent performance, particularly under the power-law distribution. The current settings for step size and simulation count in \textsc{RWB} are insufficient to accurately estimate the quantities, leading to significantly poorer results. This highlights the limitations of \textsc{RWB} in handling complex distributions and further underscores the superiority of our proposed algorithms.
\begin{figure}[t]
    \centering
    \includegraphics[width=\linewidth]{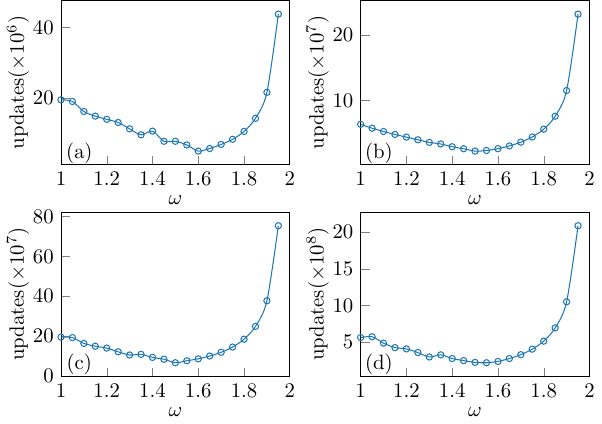}
    \caption{Impact of different values of relaxation factor $\omega$ on the number of updates performed by the \textsc{BLISOR}. The experiments are conducted on four real-world datasets: (a) Email-Enron, (b) DBLP, (c) YoutubeSnap and (d) Flixster. }
    \label{fig:wpara}
\end{figure}

\begin{table*}[t]
    \caption{Statistics of real-world directed networks used in our experiments and comparison of running time (seconds) between \textsc{FS} (\textsc{ForestSample}), \textsc{BLI} and \textsc{BLISOR}. \textsc{FS} are evaluated only under uniform distribution (Unif), while \textsc{BLI} and \textsc{BLISOR} are tested under Unif, Exp (exponential distribution), and Pareto (power-law distribution).} 
    \label{tab:direct-running-times}
    \centering
    \begin{tabular}{cccc c ccc ccc c}
    \toprule
        \multicolumn{4}{c}{Dataset}
        & \multicolumn{7}{c}{Time}
        & \multirow{3}{*}{$\omega$}\\
    \cmidrule(lr){1-4}
    \cmidrule(lr){5-11}
        \multirow{2}{*}{name} & \multirow{2}{*}{$n$} & \multirow{2}{*}{$m$} & \multirow{2}{*}{$d_{\max}$}
        & \multirow{2}{*}{\textsc{FS}}
        & \multicolumn{3}{c}{\textsc{BLI}}
        & \multicolumn{3}{c}{\textsc{BLISOR}} & \\
    \cmidrule(lr){6-8}
    \cmidrule(lr){9-11}
    &&&&& Unif & Exp & Pareto
    & Unif & Exp & Pareto &\\
    
    \midrule
    
    Wiki-Vote & 7,115 & 103,689 & 893
    & 0.33 
    & 3e-3 & 4e-3 & 4e-3
    & 2e-3 & 2e-3 & 3e-3
    & 1.2\\
    DBLP & 12,590 & 49,744 & 617
    & 0.53
    & 1e-3 & 1e-3 & 1e-3
    & - & - & -
    & 1 \\
    Twitterlists & 33,101 & 23,370 & 238
    & 0.94
    & 1e-3 & 4e-3 & 1e-3
    & - & - & -
    & 1\\
    web-Stanford  & 281,903 & 2,312,497 & 255
    & 22.3 
    & 0.49  & 0.54 & 0.57
    & 0.26  & 0.30 & 0.31
    & 1.2 \\
    Google & 875,713 & 5,105,039 & 456 
    & 74.77 
    & 1.71 & 1.73 & 1.85
    & 0.90 & 0.91 & 0.92
    & 1.15\\
    Pokec & 1,632,803 & 30,622,564 & 8,763
    & 329.95 
    & 30.59 & 27.96 & 30.75
    & 11.32 & 10.51 & 11.09
    & 1.35 \\
    Wiki-Talk & 2,394,385 & 5,021,410 & 100,022
    & 78.59 
    & 0.18 & 0.18 & 0.22
    & - & - & -
    &1\\
    LiveJournal & 4,847,571 & 68,993,773 & 20,293
    & 873.00  
    & 64.43 & 65.72 & 71.79
    & 28.74 & 29.92 & 31.38
    & 1.3 \\
    Twitter & 41,652,230 & 1,468,365,182 & 770,155
    & 5766.83 
    & 2631.33 & 2744.95 & 2843.71
    & 1265.18 & 1243.72 & 1265.01
    & 1.3\\
    \bottomrule
    \end{tabular}
\end{table*}

\subsection{Additional Experiments in Directed Graphs} \label{appdix:exp_di}
In this subsection, we present additional experimental results related to directed graphs, and note that the theoretical guarantee about error also applies to directed graphs. We conducted experiments using 9 directed graph networks with node numbers ranging from 7115 to tens of millions.  The detailed information is shown in the first 4 columns of the Table~\ref{tab:direct-running-times}.
For our algorithms \textsc{BLI} and \textsc{BLISOR}, we use the same configuration as undirected graphs. For the baseline algorithm \textsc{ForestSample}  shown in Algorithm~\ref{alg:sampleforest}, we set the sampling quantity $l = 2000$ and show the relative error of \textsc{ForestSample} on small graphs. 

\begin{table}[tb]
    \centering
    \caption{Max relative error for estimated equilibrium opinions under Unif (uniform distribution), Exp (exponential distribution), and Pareto (power-law distribution), comparing \textsc{FS} (\textsc{ForestSample}) and \textsc{BLI}.}
    \label{tab:my_label}
    \begin{tabular}{c ccc ccc}
         \toprule
         \multirow{3}{*}{Dataset}& \multicolumn{6}{c}{Max relative error($\times 10^{-2}$)}\\
         \cmidrule(lr){2-7}
         & \multicolumn{3}{c}{\textsc{FS}} & \multicolumn{3}{c}{\textsc{BLI}}\\
         \cmidrule(lr){2-4}
         \cmidrule(lr){5-7}
         & Unif & Exp & Pareto
         & Unif & Exp & Pareto\\
         \midrule
         Wiki-Vote & 7.576 & 9.999 & 88.083
         & 0.376 & 0.435 & 0.475\\
         DBLP & 5.083 & 9.008 & 45.125
         & 0.229 & 0.184 & 0.427\\
         Twitterlists & 5.057 & 6.335 & 57.644
         & 0.114 & 0.079 & 0.069\\
         \bottomrule
    \end{tabular}
\end{table}
    \textbf{Efficiency.} We present the statistics of real-world directed networks used in our experiments and compare the running time of \textsc{ForestSample}, \textsc{BLI}, and \textsc{BLISOR}. The results, summarized in Table~\ref{tab:direct-running-times}, demonstrate that \textsc{BLI} and \textsc{BLISOR} significantly outperform \textsc{ForestSample} in terms of computational efficiency. Notably, when $\omega = 1$, \textsc{BLISOR} reduces to \textsc{BLI}, and thus the results for \textsc{BLISOR} are not shown separately in this case. Across all tested distributions, \textsc{BLI} and \textsc{BLISOR} exhibit faster running times and greater scalability compared to \textsc{ForestSample}, highlighting their superior performance in handling large-scale directed networks. These results underscore the advantages of our proposed algorithms in practical applications.

    \textbf{Accuracy. } We compare the maximum relative error for the estimated equilibrium opinion between \textsc{FS} and \textsc{BLI} under three internal distributions. The results, presented in a table, show that \textsc{BLI} achieves significantly higher accuracy than \textsc{FS}. Specifically, the maximum relative error of \textsc{BLI} is consistently below $0.005$ across all distributions, demonstrating its robustness and precision. In contrast, \textsc{FS} exhibits larger errors, highlighting its limitations in accurately estimating equilibrium opinions. These findings underscore the superior performance of \textsc{BLI} in handling diverse internal distributions, making it a more reliable choice for such tasks. 
\end{document}